\documentclass[preprint]{elsarticle}

\usepackage{amsmath, amssymb}
\usepackage{amsthm}
\usepackage{mathrsfs} 
\usepackage{subfigure}

\theoremstyle{plain}
\newtheorem{thm}{Theorem}[section]
\newtheorem*{thmI}{Theorem I}
\newtheorem*{thmII}{Theorem II}

\newtheorem{lem}[thm]{Lemma}

\theoremstyle{definition}
\newtheorem{defn}{Definition}[section]

\theoremstyle{remark}
\newtheorem{rem}{Remark}

\newcommand{\Var}{{\textsc{Var}}}
\newcommand{\Form}{{\textsc{Form}}}

\newcommand{\eg}{\textit{e.g.}}
\newcommand{\ie}{\textit{i.e.}}

\newcommand{\R}{{\mathbb{R}}}

\newcommand{\luk}{\L u\-ka\-s\-ie\-wicz}

\newcommand{\remove}[1]{}

\newcommand{\Ax}{{\mathbb{A}}}
\newcommand{\Log}{{\mathscr{L}}}

\DeclareMathOperator{\conv}{\rm Conv}
\DeclareMathOperator{\range}{\rm ran}
\journal{Fuzzy Sets and Systems}

\begin{document}

\begin{frontmatter}

\title{The logical content of triangular bases of fuzzy sets in {\L}ukasiewicz infinite-valued logic}

\author[Dico]{Pietro Codara\corref{sup}}
\ead{codara@di.unimi.it}
\author[Dico]{Ottavio M.\ D'Antona}
\ead{dantona@di.unimi.it}
\author[Mate]{Vincenzo Marra\corref{cor}}
\ead{vincenzo.marra@unimi.it}

\address[Dico]{Dipartimento di Informatica, Universit\`{a}
degli Studi di Milano, via Comelico 39, I-20135 Milano, Italy}
\address[Mate]{Dipartimento di Matematica ``Federigo Enriques'', Universit\`{a}
degli Studi di Milano, via Saldini 50, I-20133 Milano, Italy}
\cortext[cor]{Corresponding author.}
\cortext[sup]{Partially supported by \emph{Dote ricerca} --- FSE, Regione Lombardia.}

\begin{abstract}Continuing to pursue a research direction that we already explored  in connection with G\"odel-Dummett logic
and Ruspini partitions, we show here that \luk{} logic is able to express the notion of \emph{pseudo-triangular basis} of fuzzy
sets, a mild weakening of the standard notion of triangular basis. \textit{En route}  to our main result we obtain an
elementary, logic-independent characterisation of triangular bases of fuzzy sets.

\end{abstract}

\begin{keyword}
\luk{} logic \sep Fuzzy sets \sep Triangular bases \sep Abstract Schauder bases \sep Axiomatisations.


\MSC[2010] 03B50 \sep 03B52
\end{keyword}

\end{frontmatter}
\tableofcontents

\section{Prologue.}\label{s:intro}
\label{s:intro}
In this paper,
by a \emph{fuzzy set} we  always mean a
function $f \colon [0,1] \to [0,1]$, with  $[0,1]\subseteq\R$,
the real unit interval.
Throughout, we fix an integer  $n>0$, and a finite, non-empty
family
\[
P=\{f_1, \ldots, f_n\}
\]
of fuzzy sets. We further always assume that each  $f_i\in P$ is a continuous function with respect to the
usual (Euclidean) topology of $[0,1]$.

\smallskip We address here the general question, what is the
logical content of the family of fuzzy sets $P$.  By way of motivation, let us think of the real unit
interval $[0,1]$  as the normalised range of values of a physical observable, say temperature.  Then
each $f_i \in P$ can be viewed as a means of assigning a \emph{truth value} to a
proposition about temperature in some many-valued logic $\mathscr{L}$.
Had one no information at all about such propositions, one would be led to identify them with propositional variables $X_i$,
subject only to the axioms of $\mathscr{L}$. Intuitively, however, \emph{the set $P$ does encode information about
$X_1, \ldots, X_n$}.
For example, consider $P=\{f_1, f_2, f_3\}$ as in Fig. \ref{fig:ruspini},
and say $f_1$, $f_2$, and $f_3$ provide truth values for
the propositions $X_1 = $ ``The temperature is low'', $X_2 = $ ``The temperature is medium'',
and $X_3 = $ ``The temperature is high'', respectively. The shape of the functions in $P$
intuitively tells us that it is never the case that the observed temperature is both low and high. More generally,
at an intuitive level it is  clear that \emph{$P$ encodes a body $B$ of knowledge about the specific application domain}
(here, about temperature). How can
we make this intuition precise?

 If  $\mathscr{L}$ has a conjunction $\wedge$
interpreted by minimum, the proposition $X_1 \wedge X_3$ has $0$ as its only possible truth
value, \ie\
it is a contradiction. The chosen set $P$ then leads us to add  \emph{extra-logical axioms}
to $\mathscr{L}$---\textit{e.g.} $\neg (X_1 \wedge X_3)$, where $\neg$ is a negation connective---in an attempt to express the
fact that one cannot observe both a high and a low temperature at the same time. More
generally, we see that \emph{$P$ implicitly  encodes a theory $\Theta_P$ over the pure logic $\mathscr{L}$}---a \emph{theory}
 being a family of formul\ae\ required
to hold, thought of as extra-logical axioms. Crucially, the theory $\Theta_P$ is determined independently of the
specifics of the available connectives. Set
\begin{align*}\tag{*}\label{t:thetap}
\Theta_P=\left\{\varphi(X_1,\ldots,X_n)\, \mid\, \varphi(f_1(x),\ldots,f_n(x))=1 \text{ for all } x \in [0,1]\right\}.
\end{align*}
(Here, $\varphi$ is a formula of $\Log$ over the variables $X_1,\ldots,X_n$, and $\varphi(r_1,\ldots,r_n)$
denotes the evaluation of $\varphi$ at $(r_1,\ldots,r_n)\in[0,1]^n$.) Under the sole assumption that $\Log$ has a sound
 $[0,1]$-valued semantics, it is easy to show that $\Theta_P$ as in (\ref{t:thetap}) is a (deductively closed) theory over $\Log$; see Lemma \ref{l:lemmino} below.

In traditional
terminology, the axioms of the logic $\Log$ (along with their deductive consequences) are to be thought of as \emph{analytic truths},
which hold true by virtue of their form alone, independently of the circumstances. Analytic truths are
the subject matter of logic proper; however, by their very nature, they carry no information about ``the world'':
whichever analytic truth one   utters about temperature,
one can equally well utter about, say, pressure. By contrast, the additional formul\ae, or extra-logical axioms, that feature
in a theory are to be thought of as \emph{synthetic truths}---assertions that are truthful only
within a specific domain of application, by virtue of properties of that domain. So, for example,
in dealing with a certain (ideal) gas one may wish to assert as a \emph{physical} (hence extra-logical) \emph{truth}
that the product of the volume and the pressure is constant at constant temperature. But there is of course no way of
deducing such a statement from the axioms of classical logic: a world in which this specific law fails is conceivable, \ie\
is logically consistent, and hence whatever the truth expressed by the law, it is a factual, or contingent, or synthetic truth. The completeness theorem
then tells us that the statement in question is not formally derivable from the axioms of classical logic, because it has a counter-model, namely, the
possible world wherein it fails. In this precise sense, \emph{logic can teach us nothing} (factual):
 good grades in logic won't help  with your physics class.

\smallskip In light of the foregoing, we  now see how to relate the two
statements:
\begin{enumerate}
\item[(S1)] $P$ determines a theory $\Theta_P$ over $\Log$, and
\item[(S2)] $P$ encodes a body $B$ of knowledge about the specific application domain.
\end{enumerate}
Indeed,
\emph{$\Theta_P$ is none other than a verbalisation of $B$}: an exposition of $B$ in formul\ae,
so to speak. But while (S1) does provide the desired clarification of the intuition (S2),
 $\Theta_P$ may end up being a mere approximation to $B$; after all,  the linguistic resources offered by $\Log$ are limited.
The differential relationship between acceleration and velocity, for example, is hardly
exactly expressible  within most formal system that go under the name of ``logics''.

\smallskip
 In this
paper we are  concerned with one instance of the general problem of making explicit
the extra-logical information implicitly encoded by $P$. In a previous paper \cite{ijar} (see also \cite{ecsqaru}), we
addressed and solved this problem in case the background logic $\Log$ is \emph{G\"odel-Dummett
\textup{(}infinite-valued propositional\textup{)} logic} \cite[Chapter 4]{hajek}, and $P$ is assumed to be
a \emph{Ruspini partition}, \ie\ such that $\sum_{i=1}^nf_i(x)=1$ for each $x\in[0,1]$. There, we
proved that G\"odel-Dummett logic can only  capture the semantical notion
of Ruspini partition up to an equivalence relation that we determined exactly. Here, we address the
problem of identifying the  synthetic, factual content of \emph{triangular bases of fuzzy sets}, a notion
strictly stronger than Ruspini partitions. Such triangular bases
 commonly occur in applications. The set $P=\{f_1,f_2,f_3\}$ provides an example;
please see Definition \ref{def:triangularbasis} below for details. Throughout the paper, we will
take $\Log$ to be \emph{{\L}ukasiewicz \textup{(}infinite-valued propositional\textup{)} logic}
 \cite{cdm}; background is provided in Section \ref{ss:luk}. It is well known that {\L}ukasiewicz logic
 is able to express addition of real numbers exactly, and so does axiomatise Ruspini partitions
 exactly. We shall prove in Theorem II of Section \ref{s:axiom} the stronger result that {\L}ukasiewicz logic axiomatises the notion of
triangular bases of fuzzy sets almost exactly; specifically, the logic axiomatises \emph{pseudo-triangular bases}
(see Definition \ref{def:triangularbasis}), a mild weakening of triangular bases. It will
transpire that the reason why the latter cannot be characterised exactly is that the logic
does not express \emph{\textup{(}affine\textup{)} linearity}, though, as mentioned, it does express addition.
Besides a fair amount of standard machinery in \luk{} logic, the proof of Theorem II  uses Theorem I, which
we prove in Section \ref{s:char}. Here we characterise pseudo-triangular bases of fuzzy sets by
elementary properties of the set of functions $P$ which strengthen the Ruspini condition.
In the final Section \ref{s:further} we discuss further research, and connections with previous
work on the algebraic semantics of {\L}ukasiewicz logic.

\section{Properties of fuzzy sets.}\label{s:properties}
Fuzzy sets are often required to satisfy additional
conditions that are deemed useful for the specific application under
consideration. Here is a popular one that we already mentioned in the Prologue, and is usually traced back\footnote{Let us mention in passing that Ruspini partitions have been long studied in general topology, where they are known as
(\emph{finite}) {\em partitions of unity}; see \eg\ the survey \cite{partitions}, and references therein.} to \cite[p.\ 28]{ruspini}.
We say
$P$ is a \emph{Ruspini partition} if  for all $x \in [0,1]$
\begin{equation}
\label{eq:Ruspini}
\sum_{i=1}^n{f_i(x)}=1\,.
\end{equation}
We further say that $P$ is \emph{$2$-overlapping} if for all $x \in [0,1]$ and all triples
of indices $i_1\neq i_2\neq i_3$ one has
\begin{equation}\label{eq:2overlap}
\min{\{f_{i_1}(x),f_{i_2}(x),f_{i_3}(x)\}}=0 \ .
\end{equation}
Figure \ref{fig:ruspini} shows a family of fuzzy sets which is both Ruspini and
$2$-overlapping.
\begin{figure}[ht!]
        \begin{center}
                \includegraphics[scale=0.5]{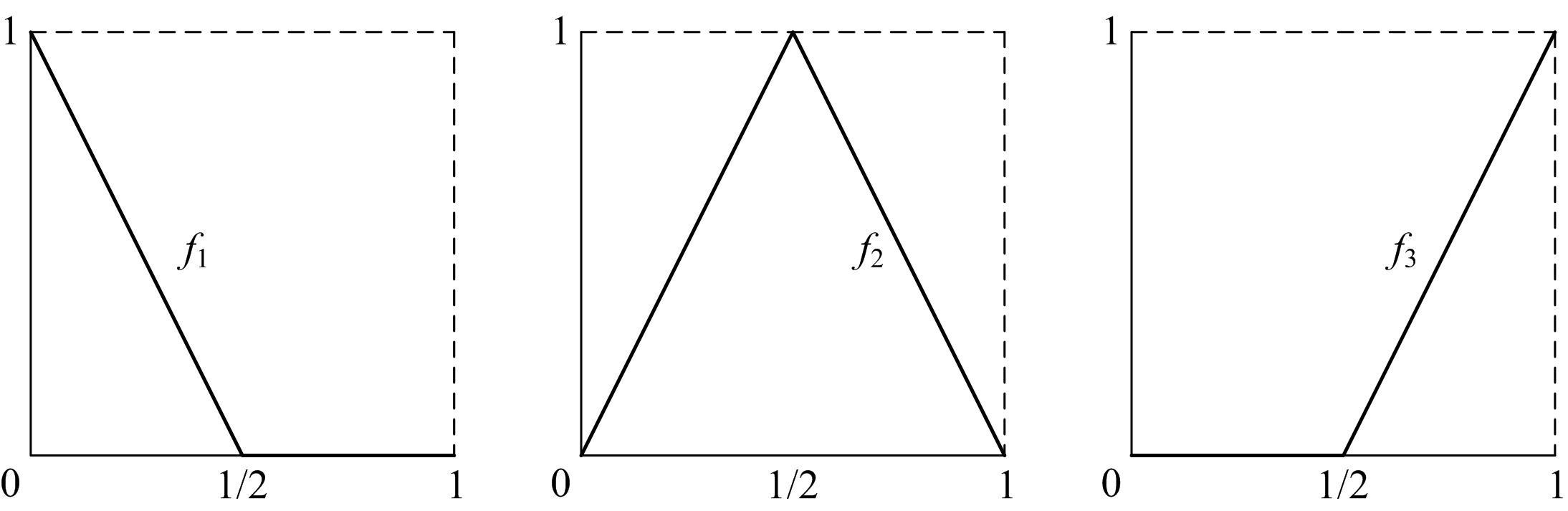}
        \end{center}
        \caption{A Ruspini and $2$-overlapping family of fuzzy sets.}
        \label{fig:ruspini}
\end{figure}

\medskip
The Ruspini and the 2-overlapping conditions (\ref{eq:Ruspini}--\ref{eq:2overlap})
apply to a family of fuzzy sets.  In the literature,
several properties applicable to a single fuzzy set have been considered too.
One of these we are assuming throughout, as stated at the beginning; namely, continuity.
Further, a fuzzy set $f \colon [0,1] \to [0,1]$ is \emph{normal} if there exist
$x \in [0,1]$ such that $f(x)=1$. If, moreover, $f(y) \neq 1$ for all $y \in [0,1]$ with $y \neq
x$, we say that $f$ is \emph{strongly normal}.
The fuzzy sets $f_1$, $f_2$, and $f_3$ depicted in Figure \ref{fig:ruspini}
are strongly normal.
The last property we wish to consider is convexity.
Following \cite[p.\ 25]{prade}, it is common to consider a weaker form of convexity than the classical one.
The fuzzy set $f\colon [0,1]\to [0,1]$ is \emph{min-convex}\footnote{We adopt this terminology to avoid confusion with convexity proper.} if for
all $x,y,\lambda \in [0,1]$,
\begin{equation}\label{eq:minconvex}
f(\lambda x + (1-\lambda)y) \geq \min(f(x),f(y)),
\end{equation}
and it is \emph{strictly min-convex} if
\begin{equation}\label{eq:sminconvex}
f(\lambda x + (1-\lambda)y) > \min(f(x),f(y)).
\end{equation}
We shall make crucial use of a localised version of min-convexity in our results. Let us call $S_f=\{x \in [0,1]\ |\ f(x)>0\}$ the \emph{support} of $f$.
We say $f$ is \emph{min-convex on its support} if (\ref{eq:minconvex}) holds for each $x,y \in [0,1]$
such that $[x,y] \subseteq S_f$. We define the notion of \emph{strict min-convexity of $f$ on its support} in the same manner,
\textit{mutatis mutandis}.

A min-convex fuzzy set is shown in Figure \ref{fig:min-con}; a  non-min-convex fuzzy set is
shown in Figure \ref{fig:nmin-con}.
\begin{figure}
 \centering
 \subfigure[A strictly min-convex fuzzy set.]
   {\label{fig:min-con}\includegraphics[scale=0.7]{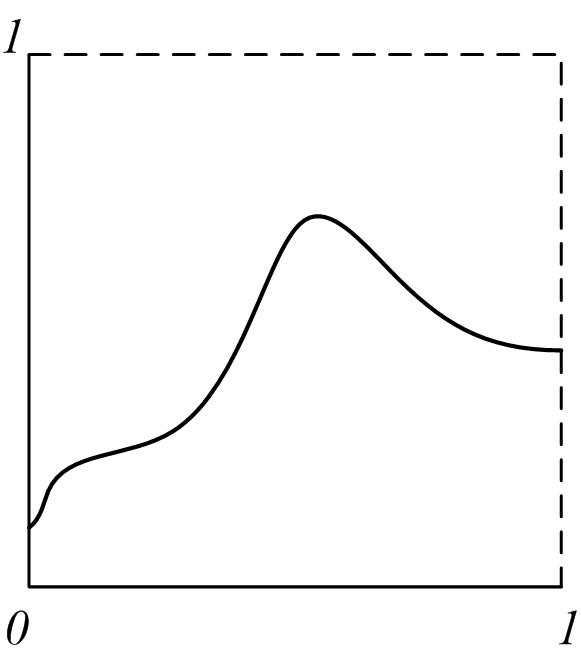}}\hspace{2cm}
 \subfigure[A non-min-convex fuzzy set.]
   {\label{fig:nmin-con}\includegraphics[scale=0.7]{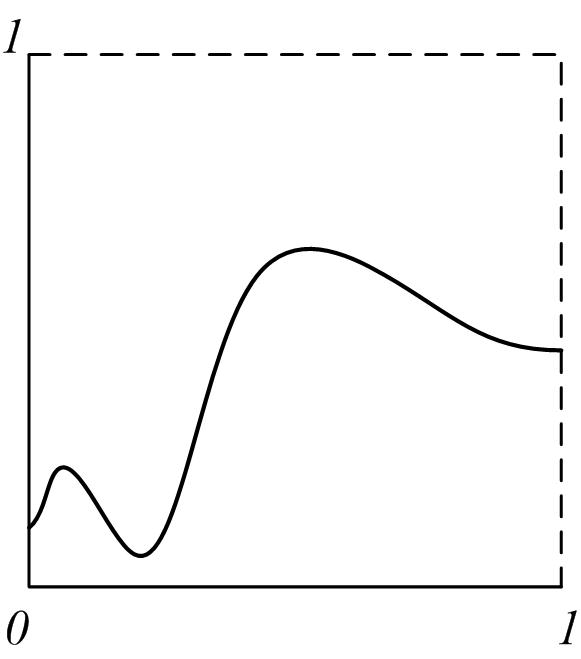}}
 \caption{Min-convexity.}
 \end{figure}

\begin{lem}\label{lem:minconvex}
A fuzzy set $f\colon [0,1]\to [0,1]$ is min-convex if, and only if,
for any $0\leq x<z<y\leq 1$ we have that
$$\mbox{if }f(z) < f(x)\mbox{ then }f(y) \leq f(z)\,.$$
Moreover, $f$ is strictly min-convex if, and only if,
for any $0\leq x<z<y\leq 1$ we have that
$$\mbox{if }f(z) \leq f(x)\mbox{ then }f(y) < f(z)\,.$$
\end{lem}
\begin{proof}
This is a straightforward verification.
\end{proof}

There is one last property of fuzzy sets that we consider  in this paper:
\begin{defn}\label{d:separating}
A finite family $P =\{f_1,\ldots,f_n\}$ of  fuzzy sets
is \emph{separating} if for all $x,y \in [0,1]$, with $x \neq y$,
$\{f_1(x),\ldots,f_n(x)\}\neq\{f_1(y),\ldots,f_n(y)\}$.
\end{defn}
\noindent It may be remarked that many families of fuzzy sets that have been investigated in the literature, or have been used in implementations, indeed are separating --- the set in Figure \ref{fig:ruspini} being a typical instance. We shall see in due course that the property is a crucial feature of such examples, \textit{cf.}\ Theorem II below.

\medskip Instead of asking that $P$ (or its members) satisfy a given
general property such as the
ones above, we can decide to restrict the choice of fuzzy sets to a
prototypical class of functions. So, for example,
a fuzzy system might use  sigmoid, or triangular, or trapezoidal
functions only. In the case of triangular functions, moreover, it is
 common to require that the various fuzzy sets fit together
nicely, as in the following definition that is
central to our paper.

\medskip
\begin{defn}\label{def:triangularbasis}A finite family $P =
\{f_1,\ldots,f_n\}$ of continuous fuzzy sets is
a \emph{pseudo-triangular basis} if
there exist $0=t_1 < t_2 < \cdots < t_{n-1} < t_n = 1$ such that
\textup{(}up to
a permutation of the indices\textup{)} for each
$i = 1, \ldots, n-1$
\begin{enumerate}
\item[\textit{a})] $f_i(t_i)=1$, $f_i(t_{i+1})=0$,
\item[\textit{b})] $f_j(x)=0$, for $x \in [t_i,t_{i+1}]$,
$j \neq i,i+1$,
\item[\textit{c})] $f_{i+1}(x)=1 - f_{i}(x)$, for
$x \in [t_i,t_{i+1}]$, and
\item[\textit{d})] $f_i, f_{i+1}$ are bijective when restricted to
$[t_i,t_{i+1}]$.
\end{enumerate}
Further, $P$ is a \emph{triangular basis} if the following
condition holds
in place of \textit{d}).
\begin{enumerate}
\item[\textit{d}$^*$)] $f_i, f_{i+1}$ are linear over
    $[t_i,t_{i+1}]$.
\end{enumerate}
\end{defn}
See Figure \ref{fig:pseudo}  for an example of a pseudo-triangular basis of fuzzy sets.
\begin{figure}[ht!]
        \begin{center}
                \includegraphics[scale=0.5]{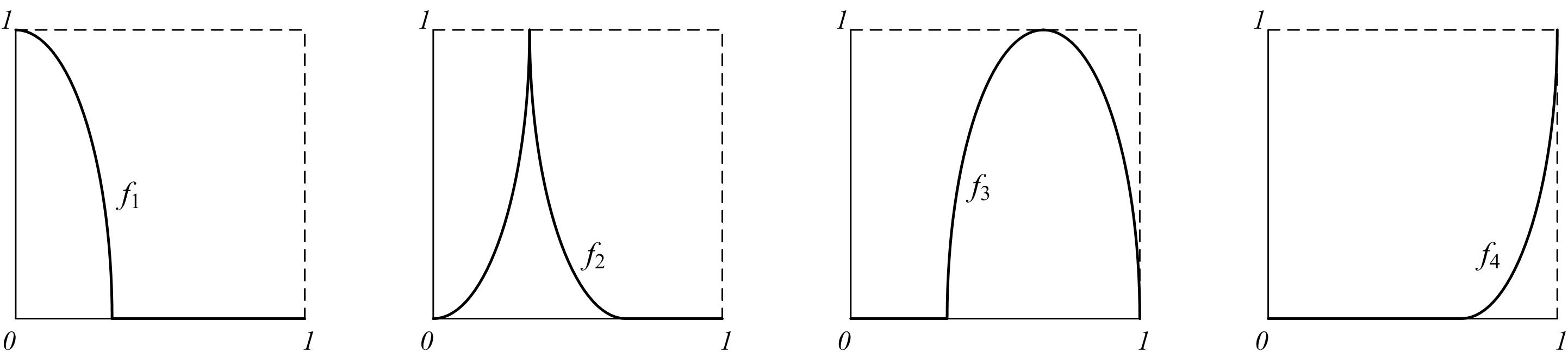}
        \end{center}
        \caption{A pseudo-triangular basis.}
        \label{fig:pseudo}
\end{figure}

\medskip \begin{rem}
It is straightforward to prove that a finite family $\{f_1,\ldots,f_n\}$
of continuous fuzzy sets is
a  triangular basis if, and only if,
there exist real numbers $0=t_1 < t_2 < \cdots < t_{n-1} < t_n = 1$
such that \textup{(}up to
a permutation of the indices\textup{)} for each
$i =1,2,\ldots, n$,
\begin{enumerate}
\item[\textit{i})] $f_i(t_i)=1$,
\item[\textit{ii})] $f_i(t_j)=0$, for $j\neq i$, and
\item[\textit{iii})] $f_i$ is linear on each interval
    $[t_k,t_{k+1}]$,
$k=1,\ldots,n-1$.
\end{enumerate}
The conditions in Definition \ref{def:triangularbasis} are somewhat more involved in order to capture instances that are not locally linear, as in Figure \ref{fig:pseudo}.
\end{rem}

\section{Characterisation of pseudo-triangular bases of fuzzy sets.}\label{s:char}

We define a continuous map
\[
T_{P}\colon [0,1] \to [0,1]^n
\]
associated with $P$ by
$$t \mapsto (f_1(t), \ldots, f_n(t))\,.$$
We write $\range T_{P} = T_{P}([0,1])$ for the range of $T_{P}$.

\medskip Recall\footnote{For background on the few basic notions from piecewise linear
geometry we use here, please see \cite{rourke}.} that the \emph{fundamental simplex} in $\R^n$, denoted by
$\Delta_n$, is the convex hull of the standard basis of $\R^n$; the
latter is denoted $\{e_1,\ldots,e_n\}$.
In symbols,
$$\Delta_n = \conv{\{e_1,\ldots,e_n\}}\,.$$
A \emph{face} of dimension $k$ of $\Delta_n$
is a subset $\conv{\{e_{i_1},\ldots,e_{i_{k+1}}\}} \subseteq \Delta_n$,
for $1\leq i_1 < i_2 < \cdots < i_{k+1} \leq n$.
A \emph{vertex} is a $0$-dimensional face.
The \emph{$1$-skeleton} of $\Delta_n$, written $\Delta_n^{(1)}$, is the
collection
of all faces of $\Delta_n$ having dimension not greater than
$1$.\footnote{Thus, $\Delta_n^{(1)}$ happens to be a graph.}

\medskip
We say $\range T_{P}$ is a \emph{Hamiltonian path} if there is a permutation
$\pi\colon\{1,\ldots,n\}\to\{1,\ldots,n\}$ such that
\begin{equation}\label{eq:hamiltoniapath}
\range T_{P} = \bigcup^{n-1}_{i=1} \conv{\{e_{\pi(i)},e_{\pi(i+1)}\}}
\end{equation}

\begin{thmI}\label{th:1}
The following are equivalent.
\begin{enumerate}
\item[i\textup{)}] $P$ is a pseudo-triangular basis.
\item[ii\textup{)}] $P$ is a $2$-overlapping Ruspini partition and each
$f_i\in P$ is strongly normal, min-convex, and strictly min-convex on its support.
\item[iii\textup{)}] The map $T_{P}: [0,1] \to [0,1]^n$ is injective, and $\range T_{P}$ is
a Hamiltonian path on $\Delta_n^{(1)}$.
\end{enumerate}
\end{thmI}

\begin{proof}
Labels \textit{a}), \textit{b}),
\textit{c}), and \textit{d})  in this proof  refer to the items in Definition \ref{def:triangularbasis}.

\medskip
\textit{i}) $\Rightarrow$ \textit{iii}).
By  \textit{b}) and \textit{c}), we immediately obtain that
\begin{align}\label{eq:Theta_in_Delta}
\range T_{P} \subseteq \Delta_n^{(1)}\,.
\end{align}
By \textit{a}), there exist
$0=t_1<t_2<\cdots<t_n=1$ such that,
up to a permutation of the indices,
\begin{align}\label{eq:T(t_i)}
T_{P}(t_i)=e_i\,,\ \mbox{for each}\ i =1,\ldots,n\,.
\end{align}
Let us fix an interval $[t_i,t_{i+1}]$, for some $i \in \{1,\ldots,n-1\}$. By (\ref{eq:Theta_in_Delta}--\ref{eq:T(t_i)}) and \textit{b}), $T_{P}([t_i,t_{i+1}]) \subseteq \conv{\{e_i,e_{i+1}\}}$.
Again by (\ref{eq:T(t_i)}), since $T_{P}$ is continuous, using the intermediate value theorem we obtain
\begin{align}\label{eq:T=conv}
T_{P}([t_i,t_{i+1}]) = \conv{\{e_i,e_{i+1}\}}\,.
\end{align}
Thus,
\[
\range T_{P} = \bigcup^{n-1}_{i=1} T([t_i,t_{i+1}]) = \bigcup^{n-1}_{i=1} \conv{\{e_{i},e_{i+1}\}}\,,
\]
that is, $\range T_{P}$ is a Hamiltonian path.

\smallskip
It remains to show that $T_{P}$ is injective. If not (\textit{absurdum hypothesis}),
there exist $x,y\in [0,1]$, with $x < y$ such that $T_{P}(x)=T_{P}(y)$.
Then (\ref{eq:T=conv}) entails that  $x,y \in [t_i,t_{i+1}]$, for some $i$.
But the fact
that $f_i(x)=f_i(y)$ and $f_{i+1}(x)=f_{i+1}(y)$, for $x\neq y$, contradicts
\textit{d}).

\medskip
\textit{iii}) $\Rightarrow$ \textit{ii}).
Since $\range T_{P} \subseteq \Delta_n$, we have $\sum_{i=1}^n{f_i(x)}=1$
for all $x \in [0,1]$, that is, $P$ is a Ruspini partition.
Since $\range T_{P} \subseteq
\Delta_n^{(1)}$, $(f_1(x),\ldots,f_n(x))$ has at most $2$ non-zero
coordinates, for each $x \in [0,1]$, that is, $P$ is $2$-overlapping.
By the definition of Hamiltonian path,
$\range T_{P}$ contains all vertices of $\Delta_n^{(1)}$.
Thus, each $f_i$ is normal. Since, moreover, $T_{P}$ is injective,
each $f_i$ is strongly normal.

\smallskip
Up to a permutation
of the indices, there exist $t_1 < t_2 < \cdots < t_{n-1} < t_n$, such
that $T_{P}(t_i)=e_i$, for each $i=1,\ldots,n$. Moreover, by the intermediate value theorem we
have $T_{P}([t_{i},t_{i+1}])\supseteq\conv{\{e_{i},e_{i+1}\}}$, $i=1,\ldots,n-1$. But since $T_{P}$ is injective it follows at once that
\begin{equation}\label{eq:intermediate}
T_{P}([t_{i},t_{i+1}])=\conv{\{e_{i},e_{i+1}\}}, \ i=1,\ldots,n-1.
\end{equation}
Now (\ref{eq:intermediate}) implies, for each $i\in\{2,\dots,n-1\}$,
\begin{align}
\label{eq:th1.1}f_i(x)<f_i(y)\,, \text{ for } t_{i-1} \leq x < y \leq t_i\,;\\
\label{eq:th1.2}f_i(x)>f_i(y)\,, \text{ for } t_{i} \leq x < y \leq t_{i+1}\,;\\
\label{eq:th1.3}f_i(x)=0\,, \text{ for } x \leq t_{i-1} \text{ or } x \geq t_{i+1}\,.
\end{align}
Indeed, (\ref{eq:th1.3}) is immediate, and  (\ref{eq:th1.1}) follows from the injectivity of $T_{P}$: if $f_{i}(x)=f_{i}(y)$ then $f_{i-1}(x)=1-f_{i}(x)=1-f_{i}(y)=f_{i-1}(y)$, so that $T_{P}(x)=T_{P}(y)$, a contradiction. The proof of (\ref{eq:th1.2}) is analogous.
Similar arguments show that, for each $0\leq x < y \leq 1$, $f_1(x)\geq f_1(y)$ and $f_n(x)\leq f_n(y)$.

We can now show that $f_{i}$ is min-convex, for each $i=1,\ldots,n$. By  Lemma \ref{lem:minconvex}
it suffices to show  that whenever  $0\leq x < y < z \leq 1$, and
$f_i(x)>f_i(y)$, then
$f_i(y)\geq f_i(z)$. The cases $i=1$ and $i=n$ are trivial;
assume  $1<i<n$. Since $f_i(x)>0$, we have $t_{i-1} < x < t_{i+1}$ by (\ref{eq:th1.3}).
If $x \geq t_{i}$, by (\ref{eq:th1.2}) and (\ref{eq:th1.3}),
$f_i(y)\geq f_i(z)$. If $x < t_{i}$, then, by (\ref{eq:th1.1}),
$y > t_{i}$. Using (\ref{eq:th1.2}--\ref{eq:th1.3}), we obtain
$f_i(y)\geq f_i(z)$. In each case, if $f_i(x)>f_i(y)$, then
$f_i(y)\geq f_i(z)$. A similar argument using (\ref{eq:th1.1}--\ref{eq:th1.3}) and Lemma \ref{lem:minconvex}
 proves that each $f_i$ is strictly
min-convex on its support.

\bigskip
\textit{ii}) $\Rightarrow$ \textit{i}).
Since each $f_i$ is strongly normal, and $P$ is Ruspini, there exist
$0 \leq t_1 < t_2 < \cdots < t_{n-1} < t_n \leq 1$ such that (up to
a permutation of the indices) for each
$i =1,\ldots,n$ we have
\begin{align}\label{eq:f_i}
f_i(t_i)=1,\ f_i(t_j)=0,\ \mbox{for}\ j \neq i\,.
\end{align}
Moreover, $t_1=0, t_n=1$. For suppose $t_1 >0$ (\textit{absurdum hypothesis}). Then  $f_i(0)<1$ for each $i =1,\ldots,n$. Since $\sum_{i=1}^n{f_i(0)}=1$, and since $P$ is $2$-overlapping,
there are
 exactly two indices $h>k \in \{1,\dots,n\}$ such that $f_h(0),$ $f_k(0)>0$.
Moreover, since $h > 1$, by (\ref{eq:f_i}) we have $f_h(t_1)=0$ and
$f_h(t_h)=1$. By Lemma \ref{lem:minconvex},
we conclude that $f_h$ is not min-convex, a contradiction. Thus $t_1 = 0$.
A similar argument shows  $t_n=1$. Summing up, there exist $0 = t_1 < t_2 < \cdots < t_{n-1} < t_n= 1$ such that
(\ref{eq:f_i}) holds. It immediately follows that \textit{a}) holds, too.
In order to prove \textit{b}), \textit{c}), and \textit{d})
let us fix an interval $[t_i,t_i+1]$, for  $i = 1, \dots, n-1$.

\medskip
To prove \textit{b}), suppose by way of contradiction that there exists $j \neq
i, i+1$ such that
$f_j(x) > 0$ for some $x \in [t_i,t_{i+1}]$. Say $j < i$.
Since, by (\ref{eq:f_i}), $x \neq t_i,t_{i+1}$,
we have that, on $t_j < t_i < x$, $f_j$ takes values $f_j(t_j)=1$, $f_j(t_i)=0$,
$f_j(x)>0$. By Lemma \ref{lem:minconvex}, $f_j$ is not min-convex, a contradiction.
The argument for $j>i$ is analogous,
and condition \textit{b}) is proved.

\medskip
From \textit{b}) and the hypothesis that $P$ is Ruspini,
we immediately obtain \textit{c}).

\medskip
It remains to prove  \textit{d}).
By (\ref{eq:f_i}),
$f_i(t_i)=f_{i+1}(t_{i+1})=1$ and $f_i(t_{i+1})=f_{i+1}(t_i)=0$.
Moreover, since $f_i$ and $f_{i+1}$ are strongly normal, and $P$ is Ruspini,
using \textit{b}) we have
\begin{equation}\label{eq:support}0<f_i(x),f_{i+1}(x)<1\,,\ \mbox{for all}\
x \in (t_i,t_{i+1})\,.\end{equation}
Since $f_i, f_{i+1}$ are continuous, by the intermediate
value theorem they are surjective when restricted to $[t_i,t_{i+1}]$. Suppose
now
that there exist $y < z \in (t_i,t_{i+1})$
such that $f_i(y)=f_i(z)$ (\textit{absurdum hypothesis}).
Observe that, by (\ref{eq:support}), $[y,z]$ is contained in the support of
$f_i$ and $f_{i+1}$.
Pick $w \in (y,z)$.
If $f_i(w) \leq f_i(y)$, then, by Lemma \ref{lem:minconvex}, $f_i$ is not
strictly min-convex
on its support, a contradiction.
If $f_i(w) > f_i(y)$, then, by
 \textit{c}),
$$f_{i+1}(w)=1-f_i(w) < 1-f_{i}(y)=f_{i+1}(y)=f_{i+1}(z)\,.$$
Thus $f_{i+1}$ is not strictly min-convex on its support,
a contradiction. Therefore, $f_i$ and $f_{i+1}$ are injective, and \textit{d})
holds.
\end{proof}

\section{{\it Intermezzo}: \luk{} logic.}\label{ss:luk}
\emph{{\L}ukasiewicz \textup{(}infinite-valued propositional\textup{)} logic} is a non-classical many-valued system going back to the 1920's, cf.\ the early survey  \cite[\S 3]{luktarski}, and its annotated  English translation  in \cite[pp.\ 38--59]{tarski}.
 The standard modern reference for  {\L}ukasiewicz logic is \cite{cdm}, while  \cite{mundicibis} deals
with topics at the frontier of current research. {\L}ukasiewicz logic can also be regarded as a member of a larger hierarchy of many-valued logics that was systematised by Petr H\'{a}jek in the late Nineties, cf.\ \cite{hajek}. Let us recall some basic notions.

\smallskip Let us fix once and for all the countably infinite set of propositional variables:
\[
\Var = \{X_1,X_2,\ldots,X_n,\ldots\}\,.
\]
Let us  write $\bot$ for the logical constant {\it falsum}, $\neg$ for the unary negation connective, and $\to$ for the binary implication connective. (Further derived connectives are introduced below.) The set $\Form$ of (well-formed) formul\ae\footnote{A set of conventions for omitting parentheses in formul\ae\ is usually adopted, and later extended to derived connectives. We do not spell the details here, as the conventions are analogous to the ones in classical logic, and are unlikely to cause confusion.} is defined exactly as in classical logic over the language $\{\bot,\neg,\to \}$.

\smallskip The {\L}ukasiewicz calculus is defined by the five\footnote{In  \cite[Chapter 4]{cdm} the language has no logical constants, and consequently (A0) does not appear as an axiom. We prefer to explicitly have $\bot$ in the language, and thus we add {\it Ex falso quodlibet} to the standard axiomatisation.} axiom schemata
\begin{itemize}
\item[(A0)] $\bot \to \alpha$ \hfill (\textit{Ex falso quodlibet}.)
\item[(A1)] $\alpha\to(\beta\to\alpha)$ \hfill(\textit{A fortiori}.)
\item[(A2)] $(\alpha\to\beta)\to((\beta\to\gamma)\to(\alpha\to\gamma))$ \hfill(Implication is transitive.)
\item[(A3)] $((\alpha\to\beta)\to\beta)\to((\beta\to\alpha)\to\alpha)$ \hfill(Disjunction is commutative.)
\item[(A4)] $(\neg\alpha\to\neg\beta)\to(\beta\to\alpha)$ \hfill(Contraposition.)
\end{itemize}
with {\it modus ponens} as the only deduction rule. Provability is  defined exactly as in classical logic; $\vdash \alpha$ means that formula $\alpha$ is provable.
We write $\mathscr{L}$ to denote {\L}ukasiewicz logic.

The logical constant \textit{verum} ($\top$), the conjunction ($\wedge$), the disjunction ($\vee$), and the biconditional ($\leftrightarrow$) are defined as in Table \ref{table:derivedconnectives}. From the definition of disjunction  one sees that (A3) indeed asserts the commutativity of disjunction.
  Other common  derived connectives are reported in the same table, with their definition.
  \begin{table}[htf]
\begin{center}\begin{tabular}{|c|c|c|c|}
\hline {\bf Notation} & {\bf Definition} & {\bf Name}  & {\bf Idempotent} \\
\hline\hline
$\bot$& -- & {\it Falsum} &--\\
\hline
$\top$& $\neg \bot$  & {\it Verum}   &--\\
\hline
$\neg \alpha$ & -- & Negation  & --\\
\hline$\alpha\to \beta$ & -- & Implication  &--\\
\hline
$\alpha\vee \beta$ & $(\alpha\to \beta)\to \beta$ & (Lattice) Disjunction  & Yes\\
\hline
$\alpha\wedge \beta$ & $\neg(\neg\alpha \vee \neg\beta)$ & (Lattice) Conjunction  &Yes\\
\hline
$\alpha\leftrightarrow \beta$ & $(\alpha \to \beta)\wedge (\beta \to \alpha)$  &  Biconditional  &--\\
\hline $\alpha\oplus \beta$ & $\neg \alpha\to \beta$ &  Strong disjunction & No \\
\hline $\alpha\odot \beta$  & $\neg (\alpha\to \neg \beta)$ &  Strong conjunction & No\\
\hline $\alpha\ominus \beta$  & $\neg (\alpha\to \beta)$ &  But not, or Difference & -- \\
\hline \end{tabular} \smallskip
\end{center}
\caption{Connectives in {\L}ukasiewicz logic.}\label{table:derivedconnectives}
\end{table}
  Some remarks are in order. Using the biconditional, one defines formul\ae\ $\alpha,\beta \in \Form$ to be \emph{logically equivalent} just in case $\vdash \alpha \leftrightarrow \beta$ holds. The connectives $\odot$ and $\oplus$ are then De Morgan dual: $\alpha \oplus \beta$ is logically equivalent to $\neg (\neg\alpha\odot\neg\beta)$, and $\alpha \odot \beta$ is logically equivalent to $\neg (\neg\alpha\oplus\neg\beta)$. These connectives, known as the \emph{strong disjunction} ($\oplus$) and \emph{strong conjunction} ($\odot$) of $\mathscr{L}$, play a central r\^ole both in H\'ajek's treatment of many-valued logics \cite{hajek}, and in Chang's algebraisation of $\mathscr{L}$ via \emph{MV-algebras} \cite{cdm}. They are not idempotent,
 in the sense that  $\alpha \oplus \alpha$ and $\alpha$ are not logically equivalent: only the implication $\alpha\to\alpha\oplus\alpha$ is provable; dual considerations apply to $\odot$. Conjunction ($\wedge$) and disjunction ($\vee$) also are De Morgan dual, but they are  idempotent; in fact, they are sometimes called the \emph{lattice connectives} because they induce the structure of a distributive lattice in the algebraic semantics of $\mathscr{L}$. Finally, the connective $\ominus$ is the co-implication, \ie\ the dual to $\to$.

If $S\subseteq \Form$ is any set of formul\ae, one  writes $S\vdash \alpha$ to mean that $\alpha$ is provable  in {\L}ukasiewicz logic, under the additional set of  assumptions $S$. When this is the case, one says that $\alpha$ is a \emph{syntactic consequence} of $S$. Since each one of (A0--A4) is a principle of classical reasoning, and since \textit{modus ponens} is a classically valid rule of inference,  each formula provable  in $\mathscr{L}$ is a theorem of classical propositional logic. The converse is not true: most notably, it is not hard to show that the {\it tertium non datur} law, $\alpha\vee\neg \alpha$, is not provable in {\L}ukasiewicz logic. In fact, it can be shown that the addition of $\alpha\vee\neg \alpha$ as a sixth axiom schema to (A0--A4) yields classical logic.

\smallskip By a \emph{theory} in {\L}ukasiewicz logic one means any set of formul\ae\ that is closed under provability, \ie\ is deductively closed. For any $S \subseteq \Form$,  the smallest theory that extends $S$ exists: it is the \emph{deductive closure} $S^\vdash$ of $S$, defined by $\alpha \in S^\vdash$ if, and only if, $S\vdash \alpha$.
A theory $\Theta$ is \emph{consistent} if $\Theta\not = \Form$, and \emph{inconsistent} otherwise.
A theory $\Theta$ is \emph{axiomatised by a set $S\subseteq \Form$ of formul\ae} if it so happens that $\Theta=S^\vdash$; and $\Theta$ is \emph{finitely axiomatisable} if $S$ can be chosen finite.

\smallskip
Let us now turn to the $[0,1]$-valued semantics. An \emph{atomic assignment},
or \emph{atomic evaluation}, is an arbitrary function $\overline{w}\colon \Var \to [0,1]$.  Such an atomic evaluation is uniquely extended to an \emph{evaluation} of all formul\ae, or \emph{possible world}, \ie\ to
a function $w\colon \Form \to [0,1]$, via the compositional rules:
\begin{align*}
w(\bot)&=0\,,  \\
w(\alpha\to\beta)&=\min{\{1,1-(w(\alpha)-w(\beta))\}}\,,\\
w(\neg\alpha)&=1-w(\alpha)\,.
\end{align*}
It follows by trivial computations that  the
formal semantics of derived connectives is the one reported in Table \ref{table:01connectives}.
\emph{Tautologies} are defined as those formul\ae\ that evaluate to $1$ under every evaluation.
\begin{table}[htf]
\begin{center}\begin{tabular}{|c|c|c|}
\hline {\bf Notation}    & {\bf Formal semantics} \\
\hline\hline
$\bot$& $w(\bot)=0$\\
\hline
$\top$  & $w(\top)=1$\\
\hline
$\neg \alpha$   & $w(\neg\alpha)=1-w(\alpha)$\\
\hline$\alpha\to \beta$   & $w(\alpha\to\beta)=\min{\{1,1-(w(\alpha)-w(\beta))\}}$\\
\hline
$\alpha\vee \beta$   & $w(\alpha\vee \beta)=\max{\{w(\alpha),w(\beta)\}}$\\
\hline
$\alpha\wedge \beta$   &$w(\alpha\wedge \beta)=\min{\{w(\alpha),w(\beta)\}}$\\
\hline
$\alpha\leftrightarrow \beta$ &$w(\alpha\leftrightarrow\beta)=1-|w(\alpha)-w(\beta)|$\\
\hline $\alpha\oplus \beta$  & $w(\alpha\oplus\beta)= \min{\{1,w(\alpha)+w(\beta)\}}$ \\
\hline $\alpha\odot \beta$   & $w(\alpha\odot\beta)= \max{\{0,w(\alpha)+w(\beta)-1\}}$\\
\hline $\alpha\ominus \beta$    & $w(\alpha\ominus\beta)= \max{\{0,w(\alpha)-w(\beta)\}}$ \\
\hline \end{tabular} \smallskip
\end{center}
\caption{ Formal semantics of connectives in {\L}ukasiewicz logic.}\label{table:01connectives}

\end{table}
Let us  write $\vDash \alpha$ to mean that the formula $\alpha\in \Form$ is a tautology. The relativisation of this concept to theories leads to the notion of semantic consequence. Let $S\subseteq \Form$ be any subset, and let $\Theta=S^\vdash$ be its associated theory. Given $\alpha\in \Form$, the assertion  $S\vDash \alpha$ states  that any evaluation $w\colon \Form\to [0,1]$ that satisfies $w(S)=\{1\}$ --- meaning that $w(\beta)=1$ for each $\beta \in S$ --- must also satisfy $w(\alpha)=1$. When this is the case, we say that $\alpha$ is a \emph{semantic consequence} of $S$. We write $S^{\vDash}$ for the set of semantic consequences of $S$.

\smallskip
It is an exercise to check that $\mathscr{L}$ enjoys the generalised validity theorem: for any $S \subseteq \Form$ and any $\alpha \in \Form$, if $S \vdash \alpha$ then $S \vDash \alpha$. (For a proof, see \cite[4.5.1]{cdm}.) On the other hand,   it is a  non-trivial theorem that $\mathscr{L}$ is complete\footnote{However, $\mathscr{L}$ fails strong completeness (\ie\ completeness for theories): there is a set $S\subseteq \Form$ and a formula $\alpha \in \Form$ such that $S\vDash\alpha$, but $S \not\vdash \alpha$; see \cite[4.6]{cdm}.} with respect to the many-valued semantics above: hence $\vdash \alpha$ if, and only if, $\vDash \alpha$, for any $\alpha \in \Form$. The first proof of this appeared in  \cite{rose_rosser}; see also  \cite[4.5.1 \& 4.5.2]{cdm}.

\smallskip
All of  the above can be adapted in the obvious manner to the finite set $\Var_n=\{X_1,\ldots, X_n\}$, in which case one speaks of {\L}ukasiewicz logic \emph{over $n$ \textup{(}propositional\textup{)} variables}, denoted $\mathscr{L}_n$.  The results in the sequel are formulated for $\mathscr{L}_{n}$, though they do admit extension to $\mathscr{L}$.
Although, strictly speaking, one should introduce fresh consequence relation symbols $\vdash_{n}$ and $\vDash_{n}$ for  $\mathscr{L}_{n}$, we will avoid this pedantry and use  $\vdash$ and $\vDash$ instead.  We will write $\Form_n$ for the set of   formul\ae\ whose propositional variables are contained in $\Var_n$.
\section{From functions to logic: theories induced by fuzzy sets.}\label{s:theory}
The following is a  detailed definition of $\Theta_P$ as in (\ref{t:thetap}).
\begin{defn}\label{d:thetap}(1) An assignment $\mu\colon \Form_n \to [0,1]$ is \emph{realised by $P$} (\emph{at
$x \in [0,1]$}) if $\mu(X_i)=f_i(x)$ for each $i=1,\ldots,n$.

(2) The \emph{theory $\Theta_P\subseteq \Form_n$ associated with $P$} is defined as the set of formul\ae\
$\varphi \in \Form_n$ such that $\mu(\varphi)=1$ whenever the assignment $\mu\colon \Form_n \to [0,1]$
is  realised by $P$.
\end{defn}
We record a simple fact for later use.
\begin{lem}\label{l:lemmino} The set $\Theta_P\subseteq \Form_n$ as in
Definition \ref{d:thetap} is indeed a theory, \ie\ a  deductively closed set of formul\ae.
\end{lem}
\begin{proof}For suppose $\varphi\in\Form_n$
is such that $\Theta_P\vdash\varphi$.
Since {\L}ukasiewicz logic is sound with respect to $[0,1]$-valued
assignments \cite[4.5.1]{cdm}, it follows that $\Theta_P\vDash\varphi$. If now the assignment
$\mu$ is realised by $P$, by definition it evaluates to $1$ each formula in $\Theta_P$; from $\Theta_P\vDash\varphi$ we have
$\mu(\varphi)=1$, too, and therefore $\varphi \in \Theta_{P}$. Hence $\Theta_P$ is
a theory.
\end{proof}
\begin{rem}\label{r:cons}Observe that the theory $\Theta_{P}$ in Lemma \ref{l:lemmino} is always consistent: no assignment at all $\mu\colon \Form_n \to [0,1]$
satisfies $\mu(\bot)=1$, hence $\bot\not \in \Theta_P$.\end{rem}

\begin{rem}\label{r:otherlogics}Notice that, as stated in the Prologue, Lemma \ref{l:lemmino} would hold (by the above proof) for any $[0,1]$-valued logic that satisfies the minimal requirement of soundness with respect to $[0,1]$-valued assignments. Also note that the finiteness of $P$ plays no r\^{o}le in the proof. {\it In conclusion, any
given collection of fuzzy sets gives rise to specific theory in any given $[0,1]$-valued logic.}
\end{rem}
For an application of Remark \ref{r:otherlogics} in the context of theories of vagueness, the interested reader may consult \cite{erkennt}.

\section{How to axiomatise a pseudo-triangular basis of fuzzy sets.}\label{s:axiom}
In this  section we throughout work with {\L}ukasiewicz logic over $n$ propositional variables, $\mathscr{L}_{n}$.
We prepare the following formul\ae\ in $\Form_{n}$.
\begin{align}
   \label{form:1}\rho &\,=\, X_1\oplus X_2\oplus \cdots \oplus X_n\,, &  \\
   \label{form:2}\alpha_{ij} &\,=\, \neg(X_i \odot X_j)  \,, & \text{ for $i,j=1,\ldots,n$, and $|i-j|=1$,}\\
   \label{form:3}\beta_{ij} &\,=\, \neg (X_i \wedge X_j)  \,, & \text{ for $i,j=1,\ldots,n$, and $|i-j|>1$.}
\end{align}
We further set
\begin{align*}\label{form:all}
\Ax=&\left\{\rho \right\} \cup \{ \alpha_{ij} \mid  i,j=1,\ldots, n, \text{ and } |i-j|=1\} \,\cup\\
&\cup \{\beta_{ij} \mid i,j=1,\ldots, n, \text{ and } |i-j|>1\}\,.
\end{align*}
By the \emph{$1$-set} of a formula $\varphi \in \Form_{n}$ we mean the following subset of $[0,1]^{n}$:
\begin{align*}
\left\{(x_{1},\ldots,x_{n})\in [0,1]^{n}\ \mid\ \mu_{\vec{x}}(\varphi)=1  \right\},
\end{align*}
where $\mu_{\vec{x}}\colon \Form_n \to [0,1]$ is the unique evaluation extending the assignment $X_{1}\mapsto x_{1}, \ldots, X_{n}\mapsto x_{n}$. The $1$-set of a finite set of formul\ae\ $\{\varphi_{1},\ldots,\varphi_{m}\}$, moreover, is defined to be  the $1$-set of the formula $\varphi_{1}\wedge\cdots\wedge\varphi_{m}$, or equivalently, the intersection of the $1$-sets of $\varphi_{i}$, $i=1,\ldots,m$.

\smallskip
To prove our Theorem II, the following lemma is needed.
\begin{lem}\label{lem:1-set}
The $1$-set of the set of formul\ae\ $\Ax$
is precisely the Hamiltonian path $\bigcup^{n-1}_{i=1} \conv{\{e_{i},e_{i+1}\}}$.
\end{lem}
\begin{proof}
For $n=3$ the proof is provided by Figures \ref{fig:gr1}, \ref{fig:gr2} and \ref{fig:gr3}.

\bigskip

\begin{figure}[ht!]
 \centering
 \subfigure[$1$-set of (\ref{form:3}).]
   {\label{fig:gr1}\includegraphics[scale=0.3]{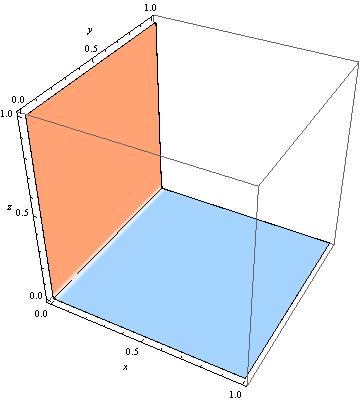}} \hspace{0.14cm}
 \subfigure[$1$-set of (\ref{form:2}--\ref{form:3}).]
   {\label{fig:gr2}\includegraphics[scale=0.3]{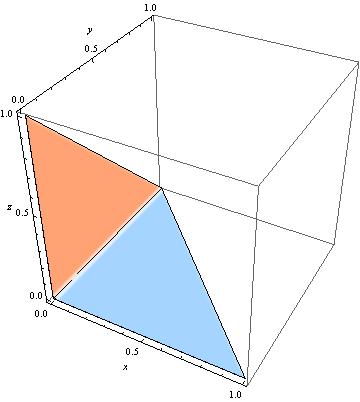}} \hspace{0.14cm}
 \subfigure[$1$-set of (\ref{form:1}--\ref{form:3}).]
   {\label{fig:gr3}\includegraphics[scale=0.3]{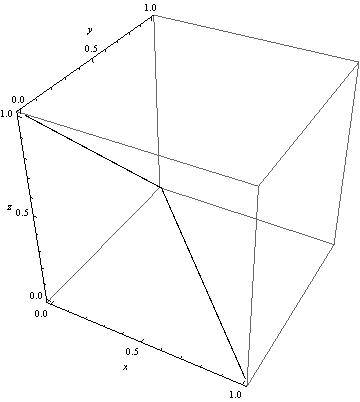}}
 \caption{The case $n=3$ of the proof of Lemma \ref{lem:1-set}.}
 \end{figure}

\bigskip In general, let $I_1$ be the $1$-set of (\ref{form:3}). Then $(x_1,\dots,x_n) \in I_1$ if, and only if, for all $i,j \in \{1,\dots,n\}$ such that $|i-j|>1$,  $1-\min\{x_i,x_j\}=1$, that is, if, and only if, one between $x_i$ and $x_j$ equals $0$. Thus, if $F_{i}$ is the $2$-dimensional face of $[0,1]^n$ containing both $e_i$ and $e_{i+1}$, we have $I_1=\bigcup^{i=1}_{n-1} F_i$.

\medskip
Let now $I_2$ be the $1$-set of (\ref{form:2}--\ref{form:3}). Then, $(x_1,\dots,x_n) \in I_2$ if, and only if, $I_2\subseteq I_1$, and, for all $i \in \{1,\dots,n-1\}$,  $1-\max\{0,x_i+x_{i+1}-1\}=1$, that is
$x_i+x_{i+1}\leq 1$. Thus, $I_2=\bigcup^{n-1}_{i=1} \conv{\{0,e_i,e_{i+1}\}}$.

\medskip
Finally, let $I_3$ be the $1$-set of (\ref{form:1}--\ref{form:3}), \ie\ of $\Ax$. Then, $(x_1,\dots,x_n) \in I_3$ if, and only if, $I_3\subseteq I_2$, and $\min\{1,x_1+\cdots+x_{n}\}=1$. Thus, $I_3=\bigcup^{n-1}_{i=1} \conv{\{e_i,e_{i+1}\}}$.
\end{proof}
\begin{thmII}
 The following are equivalent.
\begin{enumerate}
\item[i\textup{)}] $P$ is a pseudo-triangular basis of fuzzy sets.
\item[ii\textup{)}] $P$ is separating, and $\Theta_{P}=\Ax^{\vdash}$.
\end{enumerate}
\end{thmII}
\begin{proof}
{\it i}) $\Rightarrow$ {\it ii}). It is immediate to check that a pseudo-triangular
basis of fuzzy sets is separating.

\smallskip We next show that $\Ax^{\vdash}\subseteq \Theta_P$.
To this aim,
let $\mu\colon\Form_n\to [0,1]$ be an assignment realised by $P$ at
$x$. By Definition \ref{d:thetap}:
\begin{align*}
   \mu(\rho) &\,=\, \min\{1, \mu(X_1)+\cdots+\mu(X_n)\}=
   \min\{1, f_1(x)+\cdots+f_n(x)\}\,; &  \\
   \mu(\alpha_{ij}) &\,=\, 1-\max\{0,\mu(X_i)+\mu(X_j)-1\}=
     1-\max\{0,f_i(x)+f_j(x)-1\}\,, \\
     & \text{ for $i,j=1,\ldots,n$, and $|i-j|=1$;}\\
   \mu(\beta_{ij}) &\,=\, 1- \min\{\mu(X_i),\mu(X_j)\}=
    1- \min\{f_i(x),f_j(x)\}\,, \\
    & \text{ for $i,j=1,\ldots,n$, and $|i-j|>1$.}
\end{align*}
By \textit{c}) in Definition \ref{def:triangularbasis},
$\mu(\rho)=1$, and $\mu(\alpha_{ij})=1$ for all $i,j=1,\ldots,n$ such that $|i-j|=1$. By \textit{b}) in Definition \ref{def:triangularbasis},
for $|i-j|>1$, at least one between $f_i(x)$ and $f_j(x)$ equals $0$.
Thus, $\mu(\beta_{ij})=1$, for all $i,j=1,\ldots,n$, and $|i-j|>1$. Hence $\Ax\subseteq\Theta_P$, indeed.
By Lemma \ref{l:lemmino}  we therefore have $\Ax^{\vdash}\subseteq  \Theta_P$.

\smallskip It remains to prove that  $\Ax^{\vdash}\supseteq  \Theta_P$. Let $I_{\Ax}$ be the $1$-set of $\Ax$.
By Lemma \ref{lem:1-set} we have
\begin{equation}
\label{eq:step1}  {I_{\Ax} = \bigcup^{n-1}_{i=1} \conv{\{e_{i},e_{i+1}\}}\, .}
\end{equation}
On the other hand, by Theorem \ref{th:1} we have
\begin{align}
\label{eq:step2} \range{T_{P}} = \bigcup^{n-1}_{i=1} \conv{\{e_{i},e_{i+1}\}}\, .
\end{align}
Hence $I_{\Ax}=\range{T_{P}}$ by (\ref{eq:step1}--\ref{eq:step2}). If now $\varphi\in\Theta_P$, and $I_{\varphi}$ is its $1$-set,
then each assignment realised by $P$ at some point of $[0,1]$ satisfies $\varphi$ by the definition of $\Theta_{P}$, and therefore
we have
$I_{\varphi}\supseteq I_{\Ax}$. By the definition of semantic consequence we may rewrite the latter inclusion as
$\varphi \in {\Ax}^{\vDash}$. Since $\Ax$ is a finite set, by  the Hay-W\'ojcicki's Theorem \cite[4.6.7]{cdm}  we conclude $\Ax^{\vDash}=\Ax^{\vdash}$, as was to be shown.

\smallskip {\it ii}) $\Rightarrow$ {\it i}) That $P$ is separating is evidently equivalent to the fact that the map $T_{P}\colon [0,1]$ $\to [0,1]^{n}$ is injective, so let us assume the latter for the rest of this proof. Writing again $I_{\Ax}$ for the $1$-set of $\Ax$,  by Lemma \ref{lem:1-set} we have (\ref{eq:step1}). Hence it  suffices to show
\begin{equation}\label{eq:suffices}
\range{T_{P}}=I_{\Ax}\,,
\end{equation}
for then Theorem I implies that $P$ is a pseudo-triangular basis of fuzzy sets.

To prove the inclusion $\range{T_{P}}\subseteq I_{\Ax}$, let
$x=(x_{1},\ldots, x_{n})\in \range{T_{P}}$. Then the assignment $\mu(X_{i})=x_{i}$ is realised by $P$ at $x$, and thus $\mu\vDash \Theta_{P}$ by the definition of $\Theta_{P}$.  Since  $\Theta_{P}=\Ax^{\vdash}$ by assumption, and since $\Ax$ is finite, by  the Hay-W\'ojcicki's Theorem \cite[4.6.7]{cdm} we have
$\Theta_{P}= \Ax^{\vDash}$, and therefore in particular $\mu\vDash \Ax$.

To prove the converse, let us set $R=\range{T_{P}}\subseteq I_{\Ax}$. Assume by way of contradiction that $R\subset I_{\Ax}$, \ie\ there is $x \in I_{\Ax}$ such that $x \not \in R$. Therefore, if we set $D= I_{\Ax}\setminus \{x\}$, we have $R \subseteq D$. But since $R$ is the continuous image of a connected set,
namely $[0,1]$, it is itself connected, whereas  by   (\ref{eq:step1}) we see that $D$ has two connected components $D_{1}$ and $D_{n}$ containing $e_{1}$ and $e_{n}$, respectively. Hence either $R
\subseteq D_{1}$, or $R
\subseteq D_{n}$. Say the former holds, without loss of generality, so that $e_{n}
 \not \in R$. Next observe that $R$ must be a closed  set in the metric space  $I_{\Ax}$, the latter endowed with the metric $d(\cdot,\cdot)$ induced by the Euclidean distance of $[0,1]^{n}$: indeed, this is a special
 case of the well-known closed map lemma, stating that a continuous map from a compact space to a Hausdorff (in particular, metric) space must be closed (=must send closed sets to closed sets). Hence $e_{n}$ is an interior point of $I_{\Ax}\setminus R$, and thus  there is an open set $U\equiv U(e_{n},\epsilon)=\{ x \in I_{\Ax} \mid d(e_{n},x)<\epsilon\}$, for some real number $\epsilon > 0$, such that  $U\cap R=\emptyset$.
For an integer $k\geq 1$,  let us consider the formula in $\Form_{n}$
\[
\varphi_k \, =\, \underbrace{\neg X_n\oplus\cdots\oplus\neg X_{n}}_{k \text{ times}}\,.
\]
Further, let $I_{\varphi_{k}}$ be the $1$-set of $\varphi_{k}$. Direct inspection shows that $I_{\varphi_{k}}=\{(x_1,\ldots,x_{n})\in [0,1]^{n} \mid x_{n} \leq \frac{k-1}{k}\}$.  Let $k_{0}\geq 1$ be the least integer that satisfies
\[
k_{0}\geq \frac{\sqrt{2}}{\epsilon}\,.
\]
Then a simple computation shows
\begin{equation}\label{eq:terzultima}
R \subseteq I_{\varphi_{k_{0}}}\,.
\end{equation}
By (\ref{eq:terzultima}) we infer at once
\begin{equation}\label{eq:penultima}
\varphi_{k_{0}}\in \Theta_{P}\,.
\end{equation}
On the other hand, the assignment $\mu\colon \Form_{n}\to [0,1]$ such that $\mu(X_{n})=1$ and $\mu(X_{i})=0$, for $i=1,\ldots,n-1$, satisfies $\mu(\varphi_{k_{0}})=0$ and  evaluates each formula in $\Ax$ to $1$, because $\{e_{n}\}\in I_{\Ax}$. Hence
$\varphi_{k_{0}}\not \in \Ax^{\vDash}$, and therefore
\begin{equation}\label{eq:ultima}
\varphi_{k_{0}}\not \in \Ax^{\vdash}
\end{equation}
by soundness \cite[4.5.1]{cdm}.  Now (\ref{eq:penultima}--\ref{eq:ultima}) yield the desired contradiction $\Theta_{P}\not = \Ax^{\vdash}$.
\end{proof}

\section{Epilogue.}\label{s:further}
 How can we generalise the results above to situations in which
we are concerned with several physical observables? Here, we are to deal with fuzzy sets $f_i\colon [0,1]^m\to [0,1]$, $i=1,\ldots, n$,
 the integer $m\geq 1$ being the number of obser\-vables. The generalisation of triangular bases  to this setting requires
 elements of piecewise linear topology \cite{rourke}, which we assume in the following discussion. Consider a triangulation $\Sigma$ of
 $[0,1]^m$, and let $v_1,\ldots,v_l$ be the (finite) list of vertices of $\Sigma$. For each $v_i$, let $h_i\colon [0,1]^m\to [0,1]$ be the function
 such that $h_i(v_i)=1$, $h_i(v_j)=0$ if $j\neq i$, and $h_i$ agrees with an affine linear map $\R^m\to\R$ on each simplex of $\Sigma$. Then $h_i$ is automatically continuous and piecewise-linear,
 and is called the \emph{Schauder hat} of $\Sigma$ at $v_i$. The collection $H_\Sigma=\{h_i \mid i=1,\ldots, n\}$ is the \emph{Schauder basis of $\Sigma$}.
 We then define the family $P$ of fuzzy sets to be a \emph{triangular basis} if is satisfies $P=H_\Sigma$ for some triangulation $\Sigma$ of
 $[0,1]^n$. It is an exercise to check that this definition agrees with Definition \ref{def:triangularbasis} in case $n=1$. It is also easy to see
  that Schauder bases are Ruspini partitions. Unfortunately, however, no elementary characterisation of Schauder bases analogous to our Theorem
  I is known. Nonetheless, abstract characterisations of Schauder bases have been obtained using homology and other mathematical tools \cite{tams}.
  This leads to the notion of
  \emph{abstract Schauder bases}, the higher-dimensional analogue of pseudo-triangular bases of fuzzy sets, originally introduced in the last-named author's Ph.D.\ thesis. Remarkably, {\L}ukasiewicz logic
  does express the notion of abstract Schauder basis, so that it is possible to formulate a higher-dimensional analogue of our
  Theorem II. For the algebraic treatment of abstract Schauder bases in the language of lattice-groups---structures closely related to MV-algebras, the algebraic
  semantics of {\L}ukasiewicz logic---the interested reader is referred to \cite{tams, MaMu, tams2}, and to the references therein. For an account of bases
  in the context of MV-algebras and \luk{} logic themselves, see \cite{mundicibis}.

\smallskip  \luk{}  and G\"odel-Dummett logics are part of a  hierarchy of systems based on  \emph{triangular norms}; see \cite{hajek}.
It has been argued that the hierarchy, together with its generalisations, provides a framework that makes precise the notion of \emph{mathematical fuzzy logic} \cite{hb1, hb2}. The proof of Lemma \ref{l:lemmino} above, though easy, does say that the programme of axiomatising properties of fuzzy sets by means of a $[0,1]$-valued logic makes sense at a very general level.  It is important to stress that, to carry this programme out,
one needs a reasonably complete set of analogues of standard notions in mathematical logic\footnote{For the r\^{o}le that  such  notions may play even in developing a specific Mamdani-type fuzzy control system, see \cite{fn1, fn2}.}---\eg\ deductively closed theories and axiomatisations. Mathematical fuzzy logic, in the sense above, does provide such analogues. It is therefore possible, at least in principle, to develop this line of research extensively.\footnote{We already mentioned our previous contribution \cite{ijar} in this direction. Here we add in passing that it would  be important to further investigate  systems that are not based on the three standard triangular norms (\luk{}, G\"odel, and Product). For an instance of how the lack of continuity in a  triangular norm affects the formal semantics, see \cite{abm}. And for an example of how it may be more appropriate to use semantic not exclusively based on the notion of degree of truth, see \cite{abim}.}
 The benefits would surely be
equally distributed between the theoretical and the application-oriented parties. One knows more about, say, G\"odel-Dummett logic as a theoretical many-valued system, if one knows exactly to what extent the logic is capable of expressing the semantical notion of Ruspini partition. And one can make more conscious design choices in facing, say, the problem of developing a specific fuzzy-based control system, if one has that very same information about G\"odel-Dummett logic available.

\bigskip
\bigskip
\paragraph{\bf Acknowledgements} The present paper is a much-expanded follow up to the conference paper \cite{fuzzieee}.



\begin{thebibliography}{23}
\expandafter\ifx\csname natexlab\endcsname\relax\def\natexlab#1{#1}\fi
\providecommand{\bibinfo}[2]{#2}
\ifx\xfnm\relax \def\xfnm[#1]{\unskip,\space#1}\fi
\bibitem[{Aguzzoli et~al.(2009)Aguzzoli, Bianchi and Marra}]{abim}
\bibinfo{author}{S.~Aguzzoli}, \bibinfo{author}{M.~Bianchi},
  \bibinfo{author}{V.~Marra}, \bibinfo{title}{A temporal semantics for basic
  logic}, \bibinfo{journal}{Studia Logica} \bibinfo{volume}{92}
  (\bibinfo{year}{2009}) \bibinfo{pages}{147--162}.
\bibitem[{Aguzzoli et~al.(2007)Aguzzoli, Busaniche and Marra}]{abm}
\bibinfo{author}{S.~Aguzzoli}, \bibinfo{author}{M.~Busaniche},
  \bibinfo{author}{V.~Marra}, \bibinfo{title}{Spectral duality for finitely
  generated nilpotent minimum algebras, with applications},
  \bibinfo{journal}{J. Logic Comput.} \bibinfo{volume}{17}
  (\bibinfo{year}{2007}) \bibinfo{pages}{749--765}.
\bibitem[{Bova et~al.(????)Bova, Codara, Maccari and Marra}]{fn2}
\bibinfo{author}{S.~Bova}, \bibinfo{author}{P.~Codara},
  \bibinfo{author}{D.~Maccari}, \bibinfo{author}{V.~Marra}, \bibinfo{title}{A
  logical analysis of {M}amdani-type fuzzy inference, {II}: {A}n experiment on
  the technical analysis of financial markets}, in: \bibinfo{booktitle}{IEEE
  International Conference on Fuzzy Systems (FUZZ-IEEE), 2010}, pp.
  \bibinfo{pages}{1--8}. \bibinfo{note}{{10.1109/FUZZY.2010.5584834}}.
\bibitem[{Bova et~al.(2010)Bova, Codara, Maccari and Marra}]{fn1}
\bibinfo{author}{S.~Bova}, \bibinfo{author}{P.~Codara},
  \bibinfo{author}{D.~Maccari}, \bibinfo{author}{V.~Marra}, \bibinfo{title}{A
  logical analysis of {M}amdani-type fuzzy inference, {I}: {T}heoretical
  bases}, in: \bibinfo{booktitle}{IEEE International Conference on Fuzzy
  Systems (FUZZ-IEEE), 2010}, pp. \bibinfo{pages}{1--8}.
  \bibinfo{note}{{10.1109/FUZZY.2010.5584830}}.
\bibitem[{Cignoli et~al.(2000)Cignoli, D'Ottaviano and Mundici}]{cdm}
\bibinfo{author}{R.L.O. Cignoli}, \bibinfo{author}{I.M.L. D'Ottaviano},
  \bibinfo{author}{D.~Mundici}, \bibinfo{title}{Algebraic foundations of
  many-valued reasoning}, volume~\bibinfo{volume}{7} of
  \textit{\bibinfo{series}{Trends in Logic---Studia Logica Library}},
  \bibinfo{publisher}{Kluwer Academic Publishers},
  \bibinfo{address}{Dordrecht}, \bibinfo{year}{2000}.
\bibitem[{Cintula et~al.(2011{\natexlab{a}})Cintula, H\'a{}jek and
  Noguera}]{hb1}
\bibinfo{editor}{P.~Cintula}, \bibinfo{editor}{P.~H\'a{}jek},
  \bibinfo{editor}{C.~Noguera} (Eds.), \bibinfo{title}{{H}andbook of
  {M}athematical {F}uzzy {L}ogic, 1}, volume~\bibinfo{volume}{37} of
  \textit{\bibinfo{series}{Studies in Logic -- Mathematical Logic and
  Foundations}}, \bibinfo{publisher}{College Publications},
  \bibinfo{year}{2011}{\natexlab{a}}.
\bibitem[{Cintula et~al.(2011{\natexlab{b}})Cintula, H\'a{}jek and
  Noguera}]{hb2}
\bibinfo{editor}{P.~Cintula}, \bibinfo{editor}{P.~H\'a{}jek},
  \bibinfo{editor}{C.~Noguera} (Eds.), \bibinfo{title}{{H}andbook of
  {M}athematical {F}uzzy {L}ogic, 2}, volume~\bibinfo{volume}{38} of
  \textit{\bibinfo{series}{Studies in Logic -- Mathematical Logic and
  Foundations}}, \bibinfo{publisher}{College Publications},
  \bibinfo{year}{2011}{\natexlab{b}}.
\bibitem[{Codara et~al.(2007)Codara, D'Antona and Marra}]{ecsqaru}
\bibinfo{author}{P.~Codara}, \bibinfo{author}{O.M. D'Antona},
  \bibinfo{author}{V.~Marra}, \bibinfo{title}{{Best approximation of Ruspini
  partitions in G\"odel logic}}, in: \bibinfo{editor}{K.~Mellouli} (Ed.),
  \bibinfo{booktitle}{{S}ymbolic and {Q}uantitative {A}pproaches to {R}easoning
  with {U}ncertainty, {ECSQARU}, 2007}, volume \bibinfo{volume}{4724} of
  \textit{\bibinfo{series}{Lecture Notes in Computer Science ({LNAI})}}, pp.
  \bibinfo{pages}{161--172}.
\bibitem[{Codara et~al.(2009{\natexlab{a}})Codara, D'Antona and Marra}]{ijar}
\bibinfo{author}{P.~Codara}, \bibinfo{author}{O.M. D'Antona},
  \bibinfo{author}{V.~Marra}, \bibinfo{title}{An analysis of {R}uspini
  partitions in {G}\"odel logic}, \bibinfo{journal}{Internat. J. Approx.
  Reason.} \bibinfo{volume}{50} (\bibinfo{year}{2009}{\natexlab{a}})
  \bibinfo{pages}{825--836}.
\bibitem[{Codara et~al.(2009{\natexlab{b}})Codara, D'Antona and
  Marra}]{fuzzieee}
\bibinfo{author}{P.~Codara}, \bibinfo{author}{O.M. D'Antona},
  \bibinfo{author}{V.~Marra}, \bibinfo{title}{A characterisation of bases of
  triangular fuzzy sets}, in: \bibinfo{booktitle}{{IEEE International
  Conference on Fuzzy Systems (FUZZ-IEEE), 2009}}, pp.
  \bibinfo{pages}{604--609}.
\bibitem[{Dubois and Prade(1980)}]{prade}
\bibinfo{author}{D.~Dubois}, \bibinfo{author}{H.~Prade}, \bibinfo{title}{Fuzzy
  sets and systems: theory and applications}, volume \bibinfo{volume}{144} of
  \textit{\bibinfo{series}{Mathematics in Science and Engineering}},
  \bibinfo{publisher}{Academic Press Inc. [Harcourt Brace Jovanovich
  Publishers]}, \bibinfo{address}{New York}, \bibinfo{year}{1980}.
\bibitem[{Dydak(2003)}]{partitions}
\bibinfo{author}{J.~Dydak}, \bibinfo{title}{Partitions of unity}, in:
  \bibinfo{booktitle}{Proceedings of the {S}pring {T}opology and {D}ynamical
  {S}ystems {C}onference, 2003}, volume \bibinfo{volume}{27 (1)} of
  \textit{\bibinfo{series}{Topology Proceedings}}, pp.
  \bibinfo{pages}{125--171}.
\bibitem[{H{\'a}jek(1998)}]{hajek}
\bibinfo{author}{P.~H{\'a}jek}, \bibinfo{title}{Metamathematics of fuzzy
  logic}, volume~\bibinfo{volume}{4} of \textit{\bibinfo{series}{Trends in
  Logic---Studia Logica Library}}, \bibinfo{publisher}{Kluwer Academic
  Publishers}, \bibinfo{address}{Dordrecht}, \bibinfo{year}{1998}.
\bibitem[{{\L}ukasiewicz and Tarski(1930)}]{luktarski}
\bibinfo{author}{J.~{\L}ukasiewicz}, \bibinfo{author}{A.~Tarski},
  \bibinfo{title}{{Untersuchngen \"uber den Aussagenkalk\"ul.}},
  \bibinfo{journal}{C. R. Soc. Sc. Varsovie} \bibinfo{volume}{23}
  (\bibinfo{year}{1930}) \bibinfo{pages}{30--50}.
\bibitem[{Manara et~al.(2007)Manara, Marra and Mundici}]{tams}
\bibinfo{author}{C.~Manara}, \bibinfo{author}{V.~Marra},
  \bibinfo{author}{D.~Mundici}, \bibinfo{title}{Lattice-ordered abelian groups
  and {S}chauder bases of unimodular fans}, \bibinfo{journal}{Trans. Amer.
  Math. Soc.} \bibinfo{volume}{359} (\bibinfo{year}{2007})
  \bibinfo{pages}{1593--1604 (electronic)}.
\bibitem[{Marra(2011)}]{erkennt}
\bibinfo{author}{V.~Marra}, \bibinfo{title}{The problem of artificial precision
  in theories of vagueness: a note on the r\^{o}le of maximal consistency},
  \bibinfo{journal}{submitted to Erkenntnis}  (\bibinfo{year}{2011}).
\bibitem[{Marra(ress)}]{tams2}
\bibinfo{author}{V.~Marra}, \bibinfo{title}{Lattice-ordered abelian groups and
  {S}chauder bases of unimodular fans, {II}}, \bibinfo{journal}{Trans. Amer.
  Math. Soc}  (\bibinfo{year}{in press}).
\bibitem[{Marra and Mundici(2007)}]{MaMu}
\bibinfo{author}{V.~Marra}, \bibinfo{author}{D.~Mundici}, \bibinfo{title}{The
  {L}ebesgue state of a unital abelian lattice-ordered group},
  \bibinfo{journal}{J. Group Theory} \bibinfo{volume}{10}
  (\bibinfo{year}{2007}) \bibinfo{pages}{655--684}.
\bibitem[{Mundici(2011)}]{mundicibis}
\bibinfo{author}{D.~Mundici}, \bibinfo{title}{Advanced {{\L}}ukasiewicz
  {C}alculus and {MV}-algebras}, volume~\bibinfo{volume}{35} of
  \textit{\bibinfo{series}{Trends in Logic---Studia Logica Library}},
  \bibinfo{publisher}{Springer}, \bibinfo{address}{New York},
  \bibinfo{year}{2011}.
\bibitem[{Rose and Rosser(1958)}]{rose_rosser}
\bibinfo{author}{A.~Rose}, \bibinfo{author}{J.B. Rosser},
  \bibinfo{title}{Fragments of many-valued statement calculi},
  \bibinfo{journal}{Trans. Amer. Math. Soc.} \bibinfo{volume}{87}
  (\bibinfo{year}{1958}) \bibinfo{pages}{1--53}.
\bibitem[{Rourke and Sanderson(1982)}]{rourke}
\bibinfo{author}{C.P. Rourke}, \bibinfo{author}{B.J. Sanderson},
  \bibinfo{title}{Introduction to piecewise-linear topology}, Springer Study
  Edition, \bibinfo{publisher}{Springer-Verlag}, \bibinfo{address}{Berlin},
  \bibinfo{year}{1982}. \bibinfo{note}{Reprint}.
\bibitem[{Ruspini(1969)}]{ruspini}
\bibinfo{author}{E.H. Ruspini}, \bibinfo{title}{A new approach to clustering},
  \bibinfo{journal}{Information and Control} \bibinfo{volume}{15}
  (\bibinfo{year}{1969}) \bibinfo{pages}{22--32}.
\bibitem[{Tarski(1956)}]{tarski}
\bibinfo{author}{A.~Tarski}, \bibinfo{title}{Logic, semantics, metamathematics.
  {P}apers from 1923 to 1938}, \bibinfo{publisher}{Oxford at the Clarendon
  Press}, \bibinfo{year}{1956}. \bibinfo{note}{Translated by J. H. Woodger}.

\end{thebibliography}
\end{document}